\documentclass[11pt]{article}
\usepackage[utf8x]{inputenc}
\usepackage{amsmath,amssymb,amsthm}
\usepackage{braket}
\usepackage{mathtools}
\usepackage{graphicx, color, xcolor}
\usepackage{algorithm}
\usepackage{algpseudocode}
\usepackage{relsize}
\usepackage[mathscr]{eucal}
\usepackage{paralist}
\usepackage{caption}
\usepackage{tikz}
\usepackage{xspace}
\usetikzlibrary{backgrounds,positioning,fit,patterns, decorations.pathmorphing}
\usepackage{booktabs}
\usepackage{threeparttable}
\usepackage{ifthen}
\usepackage{url}
\usepackage{hyperref}
\usepackage[margin=1in]{geometry}

\newboolean{short}
\setboolean{short}{true}
\newcommand{\onlyShort}[1]{\ifthenelse{\boolean{short}}{#1}{}}
\newcommand{\onlyLong}[1]{\ifthenelse{\boolean{short}}{}{#1}}

\makeatletter
\def\BState{\State\hskip-\ALG@thistlm}
\makeatother

\newcommand{\defeq}{\stackrel{\mathsmaller{\mathsf{def}}}{=}}

\providecommand{\SizeOfSet}[1]{\lvert#1\rvert}
\providecommand{\abs}[1]{\lvert#1\rvert}

\newtheorem{theorem}{Theorem}
\newtheorem{lemma}[theorem]{Lemma}

\newtheorem{corollary}[theorem]{Corollary}

\newtheorem{definition}{Definition}

\newtheorem{remark}[definition]{Remark}

\DeclareMathOperator{\E}{E}
\DeclareMathOperator{\Vol}{Vol}
\DeclareMathOperator{\polylog}{polylog}

\begin{document}

\title{Distributed MST: A Smoothed Analysis}


\author{
Soumyottam Chatterjee \thanks{Department of Computer Science, University of Houston, Houston, TX 77204, USA. Email: {\tt schatterjee4@uh.edu}.}
 \and Gopal Pandurangan \thanks{Department of Computer Science, University of Houston, Houston, TX 77204, USA. Email: {\tt gopalpandurangan@gmail.com}. Research supported, in part, by NSF grants CCF-1527867, CCF-1540512,  IIS-1633720,  CCF-BSF-1717075, and US-Israel BSF award 2016419.} 
\and
Nguyen Dinh Pham \thanks{Department of Computer Science, University of Houston, Houston, TX 77204, USA. Email: {\tt aphamdn@gmail.com}.}
}





\maketitle

\begin{abstract}
    We study smoothed analysis of distributed graph algorithms, focusing on the fundamental minimum spanning tree (MST) problem. With the goal of studying the time complexity of distributed MST as a function of the ``perturbation'' of the input graph, we posit a {\em smoothing model} that is parameterized by a smoothing parameter $0 \leq \epsilon(n) \leq 1$ which  controls the amount of {\em random} edges that can be added to an input  graph $G$ per round. Informally, $\epsilon(n)$ is the probability (typically a small function of $n$, e.g., $n^{-\frac{1}{4}}$) that a random edge can be added to a node per round.
The added random edges, once they are added, can be used (only) for communication.

We show  upper and lower bounds on the time complexity of distributed MST in the above smoothing model. We present a distributed algorithm that, with high probability,\footnote{Throughout, with high probability (whp) means with probability at least $1 - n^{-c}$, for some fixed, positive constant $c$.} computes an MST and  runs in $\tilde{O}(\min\{\frac{1}{\sqrt{\epsilon(n)}} 2^{O(\sqrt{\log n})}, D + \sqrt{n}\})$ rounds\footnote{The notation $\tilde{O}$ hides a  $\polylog(n)$ factor and $\tilde{\Omega}$ hides a $\frac{1}{\polylog{(n)}}$ factor, where $n$ is the number  of nodes of the graph.} where $\epsilon$ is the smoothing parameter, $D$ is the network diameter and $n$ is the network size.
To complement our upper bound, we also show a lower bound of $\tilde{\Omega}(\min\{\frac{1}{\sqrt{\epsilon(n)}}, D+\sqrt{n}\})$. We note that the upper and lower bounds essentially match except for a multiplicative $2^{O(\sqrt{\log n})} \polylog(n)$ factor.

Our work can be considered as a first step in understanding the smoothed complexity of distributed graph algorithms.
\end{abstract}





\section{Introduction and Motivation} \label{section-introduction}

Smoothed analysis of algorithms was introduced in a seminal paper by Speilman and Teng \cite{Spielman_2004} to explain why the well-studied simplex algorithm for linear programming does well in practice, despite having an (worst-case) exponential run time in theory. The high-level idea behind the smoothed analysis of the simplex algorithm is the following:
\begin{enumerate}
    \item Perturbing the input data with a \emph{small} amount of \emph{random} noise (e.g., Gaussian noise with mean zero, parameterized by the variance of the noise), and then
    
    \item Showing that the perturbed input  can be solved efficiently by the simplex algorithm, i.e., in polynomial time. In particular, Spielman and Teng quantify the run time as a function of the perturbation; the more the perturbation (i.e., larger the variance of the noise), the faster the run time.
\end{enumerate}

Smoothed analysis is thus different from the {\em worst-case} analysis of algorithms. It is also different from the {\em average-case analysis}, which assumes a probability distribution on the set of all possible inputs. Smoothed analysis, on the other hand, is sort of a hybrid between the above two --- it considers the worst-case input and then randomly perturbs it. If even small perturbations (say, adding random noise) lead to efficient run time, then this means that the worst-case is quite sensitive to the input parameters. In practice, there will  usually be noise and thus the algorithm is likely to avoid the worst-case behavior.

In this paper, we initiate the study of {\em smoothed analysis of distributed graph algorithms}. Our paper is motivated by the work of Dinitz et al.\ \cite{Dinitz_2018} who initiated the study of smoothed analysis for {\em dynamic networks} (we refer to Section \ref{section-related-work} for more details). A main contribution of our paper is positing smoothing models in the context of {\em distributed graph algorithms} and performing analyses of the models. While many smoothing models are possible for such algorithms, a key goal is to identify models that lead to non-trivial bounds on the distributed complexity (here we focus on time complexity) of fundamental graph algorithms.

We focus on the distributed minimum spanning tree (MST) problem in synchronous $\mathcal{CONGEST}$ networks (see Section \ref{sec:model} for details on this standard distributed computing model). The worst-case time (round) complexity of distributed MST has been extensively studied for the last three decades and tight bounds are now well established (see, e.g., \cite{Pandurangan_2018_EATCS, Pandurangan_2017}). There is an optimal distributed MST algorithm (see, e.g., \cite{dnabook}) that runs in $\tilde{O}(D + \sqrt{n})$ rounds, where $D$ is the graph diameter and $n$ is the number of nodes in the network. Also, there is a (essentially) matching lower bound of $\tilde{\Omega}(D + \sqrt{n})$ rounds that applies even to randomized Monte-Carlo distributed algorithms \cite{Sarma_2012}.

The lower bound is shown by presenting a weighted graph (in particular, a family of graphs) and showing that no distributed algorithm can solve MST faster. This raises a motivating question for smoothed analysis: Is the worst-case bound specific to the choice of the weighted graph (family)? Or more precisely, is it specific to the choice of the graph topology or the edge weights or both? If small perturbations do not change the worst-case bound by too much then we can say that the lower bound is {\em robust} and, if they do we can say that the bounds are {\em fragile} \cite{Dinitz_2018}. Thus smoothed analysis can lead to a better understanding of the complexity of distributed MST by studying perturbations of the worst-case input. This is one of the motivations in studying smoothed analysis of distributed algorithms.

However, to answer the above questions, one has to first come up with a suitable smoothing model. For example, one possible smoothing model, in the spirit of Speilman and Teng's original smoothing model,  would be perturbing the edge weights of the input graph by a small amount. It is apparent that if the perturbation is quite small relative to the weights (since the weights can be well-spaced), then this does not any effect on the lower bound --- it remains $\tilde{\Omega}(D + \sqrt{n})$. Another possible model, which we explore in this paper, again in the spirit of original model but now applied to perturbing the {\em topology of the input graph},  is smoothing the {\em input graph} by adding\footnote{One can also delete edges, although we don't consider this in this paper, see Section \ref{section-computing-model-main-smoothing-model}.} a {\em small} number of {\em random edges}. While there are a few possible ways to accomplish this, we focus on a particular smoothing model described next. We will discuss other smoothing models (which can be considered variants of this model) in Section \ref{section-other-smoothing-models}. A practical motivation for this kind of smoothing, i.e., adding a small number of random edges to a given graph, is that many real-world networks might be better modeled by  graphs with some underlying structure with some amount of randomness. For example, it is well know that real-world graphs have power-law degree distribution, but they are not arbitrary (worst-case) power-law graphs but can be reasonably modeled by random graphs with power-law degree distribution \cite{powerlaw}.

We consider a {\em smoothing model} (see Section \ref{section-computing-model-main-smoothing-model}) that is parameterized by a smoothing parameter $0  \leq  \epsilon =  \epsilon(n)  \leq  1$ that controls the amount of {\em random} edges that can be added to an input graph $G = (V, E)$ per round. $\epsilon(n)$  is typically  a small function of $n$, say, $\epsilon(n) = n^{-\frac{1}{4}}$. More precisely, our smoothing model allows any node to add a random edge with probability $\epsilon(n)$ in each round; the added edges can be used for communication in later rounds. (We note that the added edges, otherwise, do not change the underlying solution with respect to $G$; e.g., for MST, the added edges have weight $\infty$ and hence don't affect the MST of $G$.)  Besides this additional feature, nodes behave as in the standard model, i.e., can communicate using edges of $G$. We formally define the model in Section \ref{section-computing-model-main-smoothing-model}. Note that nodes can as well choose \emph{not} to use this additional feature. Depending on $\epsilon(n)$, the number of random edges added per round can be small. (In Section \ref{section-other-smoothing-models}, we consider a variant of this model, which essentially gives the same bounds as the ones discussed here.)

An alternate way of thinking about our smoothing model is as follows. Assume that the graph $G$ is embedded in a congested clique. The congested clique model has been studied extensively in the distributed computing literature; see, e.g., \cite{Censor-Hillel_2019, Barenboim_2018, Parter_2018_ICALP, Ghaffari_2018_PODC, Jurdzinski_2018, Korhonen_2018, Parter_2018_DISC_coloring, Parter_2018_DISC_spanners}). A node --- besides using its incident edges in $E$ --- can also {\em choose} to use a random edge (not in $G$, but in the clique) with probability $\epsilon$ in a round to communicate (once chosen, a random edge can be used subsequently till end of computation). Note that if $\epsilon$ is small, say, for example $\epsilon = O(n^{-\frac{1}{4}})$, then the probability of adding a random edge by a node in a round is small. In particular, if $\epsilon = 0$,  then this boils down to the traditional model, i.e., working on the given graph $G$ with no additional random edges, as $\epsilon$ increases, the number of random edges increases with it. 

We note that the smoothing model is sort of a hybrid between the traditional model where communication is allowed only along the edges of an arbitrary graph $G$ and a model  where $G$ is a random graph  (e.g., Erdos-Renyi graph model \cite{icdcs18,misrandom}) or an expander (see e.g., \cite{soda12} and the references therein).
In the smoothing model we start with an arbitrary graph $G$ and add random edges (parameterized by $\epsilon$).
In Section \ref{section-computing-model-main-smoothing-model}, we further explore relationship between the smoothing model and other distributed computing models.

Our goal is to study how the distributed complexity of MST varies as a function of $\epsilon(n)$ (among other usual graph parameters such as network size, network diameter, etc.). We show upper and lower bounds on the time complexity of distributed MST in the aforementioned smoothing model. We present a distributed algorithm, which (with high probability) computes an MST and runs in
\begin{center}
$\tilde{O}(\min\{\frac{1}{\sqrt{\epsilon(n)}} 2^{O(\sqrt{\log n})}, D + \sqrt{n}\})$
     rounds,
\end{center}
where $\epsilon$ is the smoothing parameter, $D$ is the network diameter, and $n$ is the network size, i.e., the number of nodes in the network.

To complement our upper bound, we also show a lower bound of
\begin{center}
    $\tilde{\Omega}(\min\{\frac{1}{\sqrt{\epsilon}}, D + \sqrt{n}\})$.
\end{center}    
Our bounds show non-trivial dependence on the smoothing parameter $\epsilon(n)$, and the bounds are essentially match except for a $2^{O(\sqrt{\log n})}$ factor and a polylogarithmic factor.
\subsection{Related work} \label{section-related-work}

Smoothed analysis was introduced by Spielman and Teng\cite{Spielman_2004} and has since been applied for various algorithms problems in the sequential setting (see, e.g., \cite{Spielman_2009} for a survey).

The only work that we are aware of in the context of smoothed analysis of distributed algorithms is that of Dinitz et al.\ \cite{Dinitz_2018} who study smoothed analysis of distributed algorithms for {\em dynamic networks}. Their dynamic network model is a dynamic graph $\mathcal{H} = G_1, G_2, \dots $ that describes an evolving network topology, where $G_i$ is the graph at round $i$. It is assumed that all graphs in $\mathcal{H}$ share the same node set, but the edges can change with some restrictions, e.g., each graph should be connected. They define a smoothing model for a dynamic graph that is parameterized with a smoothing factor $k \in \{1,2,\dots, \binom{n}{2}\}$. To $k$-smooth a dynamic graph $\mathcal{H}$ is to replace each static graph $G_i$ in $\mathcal{H}$ with a smoothed graph $G'_i$ sampled uniformly from the space of graphs that are: (1) within {\em edit distance} $k$ of $G$, and (2) are allowed by the dynamic network model (e.g., smoothing cannot generate disconnected graph). The edit distance is the number of edge additions/deletions needed to transform one graph to another, assuming they share the same node set.

Our smoothing model can also be thought of in terms of choosing a random graph within a \emph{positive} edit distance (i.e., edges are only \emph{added} to the original input graph) where the number of random edges added is proportional to $n\epsilon(n)$ (per round or in total --- see Section \ref{section-other-smoothing-models}).

Dinitz et al.\ study three well-known problems that have strong lower bounds in dynamic network models, namely, \emph{flooding}, \emph{random walks}, and \emph{aggregation}. For each problem, they study robustness/fragility of the existing bound by studying how it improves under increasing amounts of smoothing.

\section{Our Model} \label{sec:model}

We first discuss the distributed computing model and then discuss our smoothing model.
\subsection{Distributed Computing Model}

We consider a system of $n$ nodes, represented as an undirected, connected  graph $G = (V, E)$. Each edge
$e \in E$ may have an associated weight $w(e)$, which can be represented using $O(\log n)$ bits. If there is no weight on an edge, then it can be considered to be $\infty$. Each node $u$ runs an instance of a distributed algorithm and has a unique identifier $ID_u$ of $O(\log{n})$ bits.

The computation advances in \emph{synchronous} rounds, where in every round, nodes can send messages, receive messages that were sent in the same round by neighbors in $G$, and perform some local computation.

Our algorithms work in the $\mathcal{CONGEST}$ model~\cite{dnabook}, where in each round a node can send at most one message of size $O(\log{n})$ bits via a single edge (whether the edge is in $G$ or is a smoothed edge).
\subsection{Smoothing Model} \label{section-computing-model-main-smoothing-model}

Given a (arbitrary) undirected connected graph $G(V,E)$ (throughout $n=|V|)$, the smoothing model 
 allows adding some random edges to the input graph $G$, thereby ``perturbing'' graph structure. We call this process \textit{smoothing}, where we add  a small number of random edges to the original graph. We describe the process of adding edges which is parameterized by a smoothing parameter $0\leq \epsilon = \epsilon(n) \leq 1$ as follows. The smoothing parameter (which in general is a function of $n$, the network size\footnote{We sometimes
just write $\epsilon$, understanding it to be a function of $n$.} controls
the amount of random edges that can be added per round.
Henceforth, we call this as the {\em $\epsilon$-smoothing model}.

More precisely, every node, in every round,
 with probability $\epsilon$ (the smoothing parameter) can {\em add} an edge to a {\em random} node (chosen
 uniformly at random from $V$) in the graph. Let the added random edges form the set $R$ (different from the original edge set $E$). Note that we allow multi-edges in the random edge choosing process; however, if there is more than one edge between two nodes, then only one edge (especially, if it belongs to $E$) that matters.  The added edge persists for future rounds and can be used  henceforth for communication; its weight is $\infty$.  A distributed algorithm can potentially exploit these
 additional edges to improve the time  complexity.\footnote{In this paper, we focus only on
 time complexity, but message complexity can also be relevant.}

In this work,
we only consider adding edges to the graph; one can also consider deleting edges from the original graph.
 However, for many problems such as MST, it is arguably more appropriate to (potentially) add edges. In fact, deleting edges can change the graph. Whereas, in  the  $\epsilon$-smoothing model, since the added edges to the given graph $G$ are purely communicating edges (with weight $\infty$), the MST with respect to $G$ is unchanged. In fact, the model allows us to study tradeoffs between the amount of random edges added to the efficiency of computing a solution of $G$. 
 
 As mentioned earlier, the $\epsilon$-smoothing model gives a ``smooth" tradeoff between the traditional CONGEST model where there no additional random edges ($\epsilon =0$) in $G$ (the input graph) and a model where there is a random graph embedded in $G$. In this sense, it is different from studying distributed computing on (purely) random graph models or expander graph models (e.g., \cite{icdcs18,soda12,misrandom}. We note the work of Ghaffari et al\cite{Ghaffari_2017_PODC,Ghaffari_2018_DISC} embeds a  random graph in a given graph $G$ and uses this embedding to design algorithms that depend on the mixing time of $G$. This is still the traditional CONGEST model, though we use their result in our algorithms. 
 
 As mentioned in Section \ref{section-introduction},
 we can also relate the well-studied {\em congested clique} model to the $\epsilon$-smoothing model and also give a way to understand computation tradeoffs between the traditional CONGEST model and the congested clique model. Assuming the input graph $G$ is embedded in a congested clique, the $\epsilon$ parameter controls the power to use the non-graph clique edges. If $\epsilon = 0$, then we have the traditional CONGEST model and for any $\epsilon > 0$, if we spend enough rounds, then one can throw a random edge between every pair of nodes which boils down to the congested clique. Of course, this is costly, which illustrates the power of the congested clique model (where the clique edges can be used for ``free"). Studying time and message complexity bounds in terms of $\epsilon$ can help us understand the power of the clique edges with respect to solving a problem on a given input graph.

\section{Distributed MST in the Smoothing Model}

For the sake of exposition, we first present a distributed MST algorithm that runs in 
$$\tilde{O}(\min\{\frac{1}{\epsilon} + 2^{O(\sqrt{\log n})} ,\allowbreak D + \sqrt{n}\}) \text{ rounds}.$$
Then we present an improved algorithm that runs in $$\tilde{O}(\min\{\frac{1}{\sqrt{\epsilon(n)}} 2^{O(\sqrt{\log n})}, \allowbreak D + \sqrt{n}\}) \text{ rounds}.$$ The second algorithm is  a modification of the first and its time complexity
approaches the lower bound of $\tilde{\Omega}(\min\{\frac{1}{\sqrt{\epsilon}}, \allowbreak D + \sqrt{n}\})$ shown in Section \ref{sec:lower}. Thus. up to a multiplicative factor of $2^{O(\sqrt{\log n})}\polylog n$, the bounds are tight.

We give a high-level overview of our approach of our first algorithm  before we get into the technical details. 
The algorithm can be described in two parts which are described in Sections \ref{section-constructing-an-expander} and \ref{section-constructing-an-MST} respectively.
At the outset we note that if $1/\epsilon$   is larger compared to $\tilde{O}(D+\sqrt{n})$, then we simply run
the standard time-optimal MST algorithm (\cite{dnabook}) without doing
smoothing.

\subsection{Part 1: Constructing an Expander} \label{section-constructing-an-expander}

Initially the algorithm exploits the smoothing model to add about $O(\log n)$ random edges per node. This can be accomplished as follows: each node (in parallel) tries to make a random edge selection for the (first) $\Theta(\frac{\log{n}}{\epsilon})$ rounds, where $\epsilon$ is the smoothing parameter. Since the probability of adding a random edge, i.e., a smoothing edge, is $\epsilon$ per round, it is easy to show that with high probability a node will add $\Theta(\log{n})$ random edges. Via a union bound, this holds for all nodes.

Now consider the graph $R(G)$ induced {\em only} by the smoothed (random) edges of $G$ after $\Theta(\frac{\log n}{\epsilon})$ rounds. In the following lemma \ref{lemma:conductance}, we show that $R(G)$ is a graph with $O(\log n)$ \emph{mixing time}.

The proof of this result comes from the relation of $\epsilon$-smoothing model to Erdos-Renyi random graph, which we will show next. 

\begin{lemma}
\label{lemma:smooth-gnp}
  Consider a graph $G(V,E)$ under  $\epsilon$-smoothing. 
  If we invoke smoothing for $\ell$ rounds (where $\ell \epsilon = o(n)$), then the graph induced by the smoothed edges is an Erdos-Renyi random graph $G(n, p)$ where $p = \Theta\left( \frac{l\epsilon}{n} \right)$.
\end{lemma}
\begin{proof}
  We calculate the probability of a smoothed edge between nodes $u$ and $v$. Clearly the edge is present if either $u$ or $v$ successfully adds the other end during $\ell$ steps. Hence,
  $p = 1 - \left( 1 - \frac{\epsilon}{n} \right)^{2\ell} = \Theta\left( \frac{\ell\epsilon}{n} \right)$.
\end{proof}

\begin{remark}
    The following lemma (Lemma \ref{lemma:conductance}) applies only to $R(G)$ and \emph{not necessarily} to $G \cup R(G)$.
\end{remark}

\begin{lemma} \label{lemma:conductance}
    Let $G=(V,E)$ be an arbitrary undirected graph and let $R(G) = (V,F)$ be the random graph induced (only) by the set $F$ of random (smoothed) edges after $\Theta(\frac{\log n}{\epsilon})$ rounds. Then, with high probability, $R(G)$ has mixing time $\tau_{mix}(R)=O(\log{n})$.
\end{lemma}
\begin{proof}
  Using Lemma \ref{lemma:smooth-gnp}, where $\ell=\Theta(\frac{\log n}{\epsilon})$, we have $R(G)$ is a Erdos-Renyi random graph $G(n, p = \Theta(\frac{\log n}{n})$. It is well-known that with high probability the this random graph is an expander (i.e., has constant conductance) and thus has $O(\log n)$ mixing time (see e.g., \cite{Ghaffari_2017_PODC}). 
  (We refer to the Appendix --- see Section \ref{section-proof-of-the-expansion-and-conductance-lemma} --- for an alternate, self-contained proof.)
\end{proof}

\subsection{Part 2: Constructing an MST} \label{section-constructing-an-MST}

In the second part, the algorithm uses $R(G)$ as a ``communication backbone'' to construct an MST in  $O(\log^2{n})2^{O(\sqrt{\log{n}})}$ rounds.

Our algorithm  crucially uses a routing result due to
 Ghaffari et al. \cite{Ghaffari_2018_DISC, Ghaffari_2017_PODC}
 who show, given an arbitrary  graph $G=(V,E)$, how to
 do {\em permutation} (or more generally, {\em multi-commodity}) routing
 fast.  We briefly describe the problem and the main result here and refer to \cite{Ghaffari_2018_DISC} for the details. Permutation or multi-commodity routing is defined
 as follows: given source-destination pairs of nodes $(s_i,t_i) \in V \times V$ and suppose $s_i$ wants to communicate with $t_i$ ($t_i$ does not know $s_i$ beforehand, but $s_i$ knows the ID of $t_i$). The {\em width}  of the  pairs is $W$
 if each $v \in V$ appears at most $W$ times as $s_i$ or $t_i$.  The goal is to construct a routing path $P_i$ (not necessarily simple) from $s_i$ to $t_i$ 
 such that the set of routing paths has low {\em congestion}
 and low {\em dilation}. Congestion is the maximum number of times any edge is used in all the paths. Dilation is simply the maximum length of the paths.   The main result of
 \cite{Ghaffari_2018_DISC} is that  routing paths $P_i$ with {\em low congestion and dilation} 
 can be found {\em efficiently}.
 Once such low congestion and dilation routing paths are found, using a standard trick of random delay routing, it is easy to establish
 that messages can be routed between the source and destination efficiently, i.e., proportional to congestion and dilation.

 \begin{theorem}[Efficient Routing](Theorem 8 from \cite{Ghaffari_2018_DISC}).Suppose we solve a multicommodity routing instance $\{(s_i, t_i)\}_i$  and
achieve congestion $c$ and dilation $d$. Then, in $\tilde{O}(c+d)$ rounds, every node $s_i$ can send one
$O(\log n)$-bit message to every node $t_i$, and vice versa.
\end{theorem}
 
Next, let us formally restate their routing results for ease of discussion. First we state their result
on multi-commodity routing on a random graph
$G(n,\log n)$, i.e., a random graph where each node
has $O(\log n)$ random edges (each endpoint chosen uniformly at random.)

\begin{theorem}[] \label{theorem:commodity_routing}
(Theorem 1 from \cite{Ghaffari_2018_DISC})
Consider a multicommodity routing instance of width $\tilde{O}(1)$. 
There is a multicommodity routing algorithm on a random graph $G(n, (\log{n}))$ that achieves congestion and dilation $2^{O(\sqrt{\log{n}})}$, and runs in time $2^{O(\sqrt{\log{n}})}$.
\end{theorem}

Then, by the construction of a random-graph-like hierarchy routing over a given graph $G$, we have the following result.

\begin{lemma} \label{lemma:fast-routing} (Lemma 11 from \cite{Ghaffari_2018_DISC})
On any graph $G$ with $n$ nodes and $m$ edges, we can embed
a random graph $G(m, d)$ with $d \geq 200 \log n$ into $G$ with congestion $\tilde{O}(\tau_{mix} \cdot d)$ and dilation 
$\tau_{mix}$ in time $\tilde{O}(\tau_{mix} \cdot d)$.
\end{lemma}

Using Lemma
\ref{lemma:fast-routing}, we have the following trivial corollary for permutation routing.

\begin{corollary}[Permutation Routing]\label{corollary:permutation_routing}
  Consider a graph $G=(V,E)$ and a set of $n$ point-to-point routing requests $(s_i, t_i)$, where
  $s_i, t_i$ are IDs of the corresponding source and destination. Each node of $G$ is the source and
  the destination of exactly one message. Then there is a randomized algorithm that delivers all
  messages in time $\tau_{mix}(G)2^{O(\sqrt{\log{n}})}$, w.h.p.
\end{corollary}
 
 Based on their multicommodity routing algorithm, they showed how to 
 construct an MST of $G$ in 
  $\tau_{mix}(G)2^{O(\sqrt{\log{n}})}$ rounds (Ghaffari et al. \cite{Ghaffari_2017_PODC}).
  However, this MST algorithm cannot be directly employed in our setting, i.e.,  to construct a MST of $G$, over $G \cup R(G)$.

 We briefly summarize the idea from section 4 in their paper \cite{Ghaffari_2017_PODC}, to explain
why this
algorithm is not applicable  directly in a black box manner. The main reason for this: while the algorithm of
\cite{Ghaffari_2017_PODC} operates on the graph $G$, ours operates on graph $G' = G\cup R(G)$. Applying the algorithm
directly to $G\cup R(G)$ can (in general) yield an algorithm running in $\tau_{mix}(G')2^{O(\sqrt{\log{n}})}$ rounds where
  $\tau_{mix}(G')$ is the mixing time of $G'$. Note that $\tau_{mix}(G')$ (in general) can be of the same
  order of $\tau_{mix}(G)$ (even for constant smoothing parameter $\epsilon$) in some graphs\footnote{For example
  consider two cliques of size $n/2$ connected to each other via $O(n)$ edges.}. Thus the running time
  bound does not (in general) depend on $\epsilon$ and does not give our desired bound
  of $\tilde{O}(\log n/\epsilon)$ rounds. We give more details on the approach
  of \cite{Ghaffari_2017_PODC} and then discuss our algorithm.

The approach of \cite{Ghaffari_2017_PODC}   is to modify Boruvka's algorithm \cite{Nesetril_2001}, where
the MST is built by merging tree fragments. In the beginning, each node is a fragment by itself. The fragment size grows  by merging. To ensure efficient communication within a fragment, they
maintain a virtual balanced tree for each fragment. A virtual tree is defined by virtual edges among nodes,
where virtual edges are  communication paths constructed by the routing algorithm. It follows that for
each iteration, fragment merging may increase the number of virtual edges of some node $v$
by $d_G(v)$, where $d_G(v)$ is the degree of $v$ in $G$. Since Borukva's method takes $O(\log n)$ iteration, the virtual degree of any node $v$
is at most $d_G(v) O(\log n)$. Thus virtual trees communication is feasible via commodity routing
by theorem \ref{theorem:commodity_routing} which takes $\tau_{mix}(G)2^{O(\sqrt{\log{n}})}$ rounds.
Applying the above approach directly to our setting yields only an MST algorithm running in $\tau_{mix}(G')2^{O(\sqrt{\log{n}})}$ rounds as mentioned in the last paragraph.

In our setting, to obtain an algorithm running in $\tilde{O}(\log n/\epsilon)$ rounds,
 we would like to perform routing (only) on $R(G)$ instead of $G\cup R(G)$. However, if we use the same algorithm of  \cite{Ghaffari_2017_PODC}  in $R(G)$ then during the computing
of the MST in $G$, some node $v$ may become overloaded by $d_G(v)$ virtual edges which causes
too much congestion in $R(G)$ (where each node has degree $\Theta(\log n)$ only) when $d_G(v)$ is large.
It follows that the routing is infeasible in $R(G)$ and the algorithm fails. To solve the
problem we proposed a modified algorithm that uses  \textit{aggregate routing}. 

The idea of aggregate routing is to perform permutation routing, where some destinations are the same. 
In other words assume a permutation routing problem where there are only $k \leq n$ (distinct) destinations and $t_1, \dots, t_k$ and there are $n$ sources
$s_1, \dots, s_n$. Let $C_i$ be the set of sources who have the same destination $t_i$. In aggregate routing,
we would like to {\em aggregate} the set of messages in $C_i$ and deliver it to $t_i$.
The aggregate function can be a  separable (decomposable)  function such as $min$, $max$ or $sum$. Then we can modify the permutation  routing easily as follows. For a node $u$ that is performing the routing,
suppose there are multiple messages arriving at $u$ in the same round. Then $u$ computes the aggregate of the messages belong to the same set $C_i$ (for every $1\leq i\leq k$) and forwards that message according to the routing algorithm. At the destination, 
the aggregate is computed over any received messages destined for this particular destination. With this intuition, we state the following definitions
and lemma.

\begin{definition}[$k$-Aggregate Routing]\label{definition:k_aggregate_routing}
  Consider a graph $G=(V,E)$, where nodes are divided into $k$ disjoint partitions $C_1, C_2,\dots,C_k$.
  For each partition $C_i$, there is a leader $l_i$, known to all members of $C_i$. Each node $u\in V$ has one
  message to deliver to its leader. Let $f$ be a
  separable aggregate function (such as $min$, $max$, or $sum$). The
  $k$-aggregate routing problem is to compute the aggregate $f$ over the nodes in each partition and route it to the corresponding leader of the partition.
\end{definition}
We show that $k$-aggregate routing can be solved in the same time bounds as multi-commodity routing.
\begin{lemma}\label{lemma:k_aggregate_routing}
  Consider a graph $G=(V,E)$ with an instance of $k$-aggregate routing problem.
  There is a randomized algorithm that solves the problem in time
  $\tau_{mix}(G)2^{O(\sqrt{\log{n}})}$, w.h.p.
\end{lemma}
\begin{proof}
  Let $m^u$ denote the original message at $u$. We will send the tuple $(m^u, l^u)$ where $l^u$
  is the leader of the partition that $u$ belongs to. Let $f$ be the separable aggregate function.

  Each node $v$ executes the multi-commodity routing algorithm, with this extra rule:
  In a round $t$, suppose $v$ receives multiple messages having the same destination, which is some leader $l_i$,
  $v$ computes the aggregate $f_{t,i}$ over those messages, and prepares a tuple $(f_{t,i}, l_i)$.
   $v$ then  forwards the  aggregated message to the appropriate next-hop neighbor in the  routing path for the next
  round. $v$ performs the same reduction for messages targeting
  other leaders.

  It is easy to see that this routing schema is not congested, i.e., it is as fast as the multi-commodity/permutation routing.
  Observing the local invariance of permutation routing: for each destination $u$, every node, in each round of the
  algorithm, sends out at most one message routing towards $u$. This invariance holds in our $k$-aggregate
  routing, by the above construction.

  At the end of the routing, each leader aggregates over its received messages, which is the aggregate
  in its partition.
\end{proof}

While $k$-aggregate routing can be seen as {\em upcast} \cite{dnabook}, the complementary operation
of {\em downcast}, i.e., sending a message from a source to several destinations, can also be done
efficiently as shown below.

\begin{lemma}[$k$-Aggregate Routing and Downcast]\label{lemma:k_aggregate_routing_full}
  Consider a graph $G=(V,E)$ with an instance of the $k$-aggregate routing problem. Furthermore, we
  require that every member of a partition knows the corresponding aggregate value.
  There is a randomized algorithm that solve the problem in time
  $\tau_{mix}(G)2^{O(\sqrt{\log{n}})}$, w.h.p.
\end{lemma}
\begin{proof}
  Using the algorithm in lemma \ref{lemma:k_aggregate_routing}, each node also records the source of
  the incoming messages together with the associated leader ID and round number.
  Then the routing can be reversed. Starting from the leader of each partition, it sends out the
  aggregate message with its ID, and each node reverses the aggregate message towards the matching
  sender. Hence the downcast can be accomplished in the same number of rounds as $k$-aggregate routing.
\end{proof}

We are now ready to implement the MST algorithm.

\begin{theorem}
\label{theorem:mst-upper-bound}
  Consider a weighted graph $G=(V,E)$ in the $\epsilon$-smoothing model. There exists a randomized distributed
  algorithm that finds an $MST$ of $G$ in time $\tilde{O}(\frac{1}{\epsilon} + 2^{O(\sqrt{\log n})})$, w.h.p.
\end{theorem}

\begin{proof}
  As discussed in Part 1 (Section \ref{section-constructing-an-expander}) the algorithm executes $\Theta(\frac{\log n}{\epsilon})$ rounds of random edge selection to
  construct the random graph $R(G)$ (the graph induced only by the random edges). As shown in Lemma \ref{lemma:conductance}, $R(G)$ is an expander with constant conductance and hence has $O(\log n)$ mixing time.

  We use the permutation routing result of \cite{Ghaffari_2018_DISC} to construct the routing structure on $R(G)$, which allows permutation routing in
  $\tau_{mix}(R(G))2^{O(\sqrt{\log{n}})} = O(\log n)2^{O(\sqrt{\log{n}})} = 2^{O(\sqrt{\log{n}})}$ rounds.

  Our MST algorithm is based on  the standard Gallagher-Humblet-Spira (GHS)/Boruvka  algorithm, see e.g.,
  \cite{dnabook} which is also used in \cite{Ghaffari_2017_PODC} and many other MST algorithm see e.g., \cite{Pandurangan_2017, Elkin_2017_PODC}. The main modification compared
  to the standard GHS algorithm is that growth (diameter) of fragments are controlled during merging (as in controlled GHS algorithm \cite{dnabook}).

 We summarize the algorithm here and sketch how it is implemented. 
 
 Let $T$ be the (unique) MST on $G$ (we will assume that all weights of edges of $G$ are distinct).
A \emph{MST fragment} (or simply a {\em fragment}) $F$ of $T$ is defined as a connected subgraph of
$T$, that is, $F$ is a subtree of $T$. 
An \emph{outgoing edge} of
a MST fragment is an edge in $E$ where one adjacent node to the edge
is in the fragment and the other is not. The \emph{minimum-weight
outgoing edge (MOE)} of a fragment $F$ is the edge with {\em minimum
weight}  among all outgoing edges of $F$. As an immediate
consequence of the cut property for MST, the MOE of a fragment
$F=(V_F, E_F)$ is an edge of the MST. 

The GHS algorithm operates in {\em phases} (see e.g., \cite{dnabook}). 
In the first phase,  the GHS algorithm starts
with each individual node as a fragment by itself and continues till there is only one
one fragment --- the MST. That is, at the beginning, there are
$|V|$ fragments, and at the end of the last phase, a single fragment which is the
MST. All fragments find their MOE simultaneously in parallel.

In each phase, the algorithm maintains the following invariant: Each MST fragment has a leader
and all nodes know their respective parents and children. The root of the tree will be the leader.
Initially, each node (a singleton fragment) is a root node;
subsequently each fragment will have one root (leader) node.
 Each fragment is identified by the identifier of its root --- called the fragment ID --- and
 each node in the fragment knows its fragment ID.

  We describe one phase of the GHS (whp there will be $O(\log n)$ phases as discussed below). Each fragment's  operation is coordinated by the respective fragment's root (leader). 
Each phase consists of two major operations: (1) Finding MOE of all fragments and (2) Merging fragments via their MOEs.

We first describe how to perform the first operation (finding MOE).  
Let  $F$ be
  the current set of fragments.   Each node in $V$ finds its (local) minimum outgoing edge (if any), i.e.,
  an edge to a neighbor belonging to a different fragment that is of least weight. We then execute a
  $|F|$-Aggregate Routing and Downcast, using $min$ as the aggregate function, with each node being
  the source and having its fragment leader as its destination. At the end of this step, for each fragment, every member knows the minimum outgoing edge (MOE) of the entire fragment. This MOE edge will
  be chosen for merging in the second operation (merging fragments).  Also, each node keeps the reversed routing paths for further usage.

  Once the merging (MST) edges  are identified the second operation --- merging --- is processed. In order to avoid long chains of fragments, a simple randomized trick is used.
  Each fragment chooses to be a \textit{head} or \textit{tail} with probability $1/2$. Only \textit{tail} fragments will merge if their outgoing edge points to a \textit{head} fragment. It can be shown (e.g., see \cite{Ghaffari_2017_PODC} that this merging (still) leads to a constant factor decrease in the number of fragments (on average) and hence the number of phases will be $O(\log n)$ in expectation and with high probability.
  
  We now describe how a merge can be implemented efficiently,  There will be no change in the head fragment, but
  all the tail ones will update to acknowledge the head leader as the new leader. For a tail fragment
  $T$, let $v \in T$ be the node that is making the merge, $v$ knows the ID of the head leader by communicating with its neighbor which is a member of the head fragment). $v$ routes
  this new leader ID to the current leader of $T$. This is done in parallel by permutation routing.
  The current leader of $T$ downcasts the new leader ID to all $T$ members. This is done via the saved
  reversed routing paths. The merging is now completed, and time for one iteration is the same as that of permutation routing as before, i.e., 
  $O(\log n)2^{O(\sqrt{\log{n}})}$.
  
  There are $O(\log n)$ phases and each phase can be implemented in $O(\log n)2^{O(\sqrt{\log{n}})}$ rounds and hence the total time for Part 2 is $O(\log^2 n)2^{O(\sqrt{\log{n}})}$.
  
  The total time for MST construction is the number of rounds for Part 1 pluses the number of rounds in Part 2:
 
$\Theta(\frac{\log n}{\epsilon}) + O(\log^2 n)2^{O(\sqrt{\log{n}})} =
\tilde{O}(\frac{1}{\epsilon} + 2^{O(\sqrt{\log{n}})})$.
\end{proof}

\subsection{An Improved Algorithm}
\label{section-improved-algorithm}
We now present an algorithm that is a variant of the previous algorithm and improves upon it. The time complexity of the improved algorithm approaches the lower bound (cf. Section \ref{sec:lower}).   The idea is to  use {\em controlled GHS} algorithm \cite{dnabook} to construct MST fragments of suitable size.  Then we apply the smoothing (with a smaller number of rounds compared to previous algorithm, i.e., $\tilde{O}(\frac{1}{\sqrt{\epsilon}})$ instead of $\tilde{O}(\frac{1}{\epsilon})$) to add an expander over the super-graph induced by the MST fragments where each super-node is one fragment (partition).  Then we compute the final MST  in a similar fashion to Theorem \ref{theorem:mst-upper-bound} on the super-graph.

\begin{theorem}
\label{theorem:mst-lower-bound}
  Given a weighted graph $G(V,E)$ in the $\epsilon$-smoothing model.
 Then  there exists a randomized distributed algorithm that finds an MST of $G$
  in time $\tilde{O}(\frac{1}{\sqrt{\epsilon}})2^{O(\sqrt{\log n})}$, w.h.p.
\end{theorem}

\begin{proof}
  We give the algorithm along with its analysis, as follows.

    Run $O(\frac{\log n}{\sqrt{\epsilon}})$ rounds of smoothing. Denote by $S$ the set of smoothed edges generated. Using Lemma \ref{lemma:smooth-gnp},  the probability that a smoothed edge occurs between two nodes in $G$ is  $p = \Theta(\frac{\sqrt{\epsilon} \log n}{n})$.
    
     Run \emph{controlled GHS} \cite{dnabook} for $\log{\frac{1}{\sqrt{\epsilon}}}$ phases. This takes $O(\frac{\log{n}}{\sqrt{\epsilon}})$ rounds. Every cluster (each of which is an MST fragment) will have size $\Omega(\frac{1}{\sqrt{\epsilon}})$ and diameter $O(\frac{1}{\sqrt{\epsilon}})$, and there will be $O(n\sqrt{\epsilon})$ such clusters \cite[Section 7.4]{dnabook}. We call these clusters as {\em base fragments}.
     We note that communication within a cluster (i.e., between any node of the cluster and its leader) takes $O(\frac{1}{\sqrt{\epsilon}})$ rounds.
     
     View these clusters as a set of \emph{super-nodes}, denoted by $V'$. Let $E' \subset E$ be the set of inter-super-node edges, let $S' \subset S$ be the set of inter-super-node smoothed edges. Consider two super-graphs: $G'(V',E')$ and $R'(V',S')$. It is easy to show that due to the probability $p$ of the random edges introduced by the \emph{smoothing} process, the super-graph $R'(V',S')$  is an Erdős–Rényi random graph or a $G(n', p')$-random graph, where
    \begin{equation}
        n' = O(n \sqrt{\epsilon})
    \end{equation}
    and
    \begin{equation}
        p' \geq \Omega\left (\left (\frac{1}{\sqrt{\epsilon}} \right)^2\right) p = \Omega \left( \frac{\log{(n')}}{n'} \right)
    \end{equation}
    Thus, the super-graphs $G' \cup R'$ is equivalent to the smoothed graph of our model in Section \ref{sec:model}.

     Similar to the previous algorithm, we solve the MST problem using the Boruvka's  algorithm on the super-graph $G'$, using the routing structure of Ghaffari et al.\ \cite{Ghaffari_2017_PODC, Ghaffari_2018_DISC} and the aggregate routing over $R'$. To implement the algorithm on the super-graph, we pipeline messages within all the super-nodes, i.e., inside the base fragments. There are $O(\log n)$ phases of the Boruvka's algorithm and  each phase takes
    \begin{center}
        $O(\frac{1}{\sqrt{\epsilon}}) \cdot 2^{O(\sqrt{\log n})}$ rounds.
    \end{center}    
    The extra term of $O(\frac{1}{\sqrt{\epsilon}})$ is incurred by  communication within a super-node. 
Thus, in total, the second part  takes $O(\frac{1}{\sqrt{\epsilon}} \cdot \log{n}) \cdot 2^{O(\sqrt{\log n})}$ rounds.
    
    Combining the MST edges over the super-graph $G'$, and the MST edges in each super-node, we have the MST for the original graph. Therefore, the total time complexity of the algorithm is
    \[
    O(\frac{\log{n}}{\sqrt{\epsilon}})  +  O(\frac{\log{n}}{\sqrt{\epsilon}}) \cdot 2^{O(\sqrt{\log n})}
    =
      \tilde{O}(\frac{1}{\sqrt{\epsilon}}) \cdot 2^{O(\sqrt{\log n})}.
      \]

\end{proof}

\section{Lower Bound}
\label{sec:lower}
In this section, we show the following lower bound
result on the $\epsilon$-smoothing model. We note
that $\tilde{\Omega}(D+\sqrt{n})$ is an unconditional lower bound (without smoothing) that holds even for randomized Monte-Carlo approximate MST algorithms\cite{Sarma_2012}.

\begin{theorem}[Smooth MST Lower Bound]
\label{theorem:smooth-lower-bound}
    There exists a family of graphs $\mathcal{G}$, such that, under the $\epsilon$-smoothing model, any distributed MST algorithm must incur a running time of $\tilde{\Omega}(\frac{1}{\sqrt{\epsilon}})$, in expectation.
\end{theorem}

We will prove the lower bound theorem by using the technique used in \cite{Sarma_2012}. First, we will briefly recall the lower bound poof of $\tilde{\Omega}(\sqrt{n})$ (we assume $D= O(\log n)$) without smoothing. For purposes of exposition, we simplify and slightly modify the technique in the mentioned paper, to show only the bound for exact distributed MST. Then we extend it to the smoothing model. The procedure is to establish a chain of algorithm reductions which is the same as  in \cite{Sarma_2012}, such that it relates distributed MST to a problem with a known lower bound.
The following are the chain of reductions:
\begin{itemize}
\item Set Disjointness (SD) to Distributed Set Disjointness (DSD). We first reduce the   set disjointness (SE) verification problem, a standard well-studied problem  in two-party  communication complexity to the problem of distributed set disjointness (DSD) verification. In the set disjointness problem (SD), we have two parties Alice and Bob,
who each have a $k$-bit string --- $x= (x_1,x_2,\dots,x_k)$
and $y = (y_1,y_2,\dots, y_k)$ respectively.  The goal
is to verify if the set disjointness function is defined
to be one if the inner product  $\langle x, y\rangle$ is $0$ (i.e., there is no $i$ such that $x_i=y_i=1$) and zero otherwise. The goal is to solve SD by {\em communicating
as few bits} as possible between Alice and Bob.

In the distributed set disjointness (DSD) verification, the goal is to solve SD in a given input graph $G =(V,E)$, where two distinguished nodes $s, t \in V$ have the bit vectors $x$ and $y$ respectively. In other words, instead of communicating directly as in the two-party problem (between Alice and Bob), the two nodes
$s$ and $t$ (standing respectively for Alice and Bob) have to communicate via the edges of $G$ (in the CONGEST model) to solve SD. The goal is to solve the DSD problem 
using {\em as few rounds} as possible in the CONGEST model (where only $O(\log n)$ bits per edge per round are allowed).
    \item Reduction of Distributed set disjointness (DSD) to {\em connected spanning subgraph (CSS) verification}. In the CSS problem, we want to solve a {\em graph verification} problem which can be defined as follows.  In distributed graph verification, we want to efficiently check
whether a given subgraph of a network has a specified property via a distributed algorithm.
Formally, given a  graph $G=(V,E)$,  a subgraph $H=(V,E')$ with $E'\subseteq E$, and a predicate $\Pi$,
it is required to decide whether $H$ satisfies $\Pi$ (i.e., when the algorithm terminates, every node knows whether $H$ satisfies $\Pi$).  The predicate $\Pi$ may specify statements such as ``$H$ is connected'' or ``$H$ is a spanning tree" or ``$H$ contains a cycle''. Each vertex in $G$ knows which of its incident edges (if any) belong to $H$. 

In the connected spanning subgraph (CSS) verification the goal is to
verify whether the given subgraph $H$ is connected and spans all nodes of $G$, i.e., every node in $G$ is incident to some edge in $H$. The goal is to solve
CCS in as few rounds as possible (in the CONGEST model).

    \item The last reduction is to  reduce CCS  verification to computing an MST. 
\end{itemize}

The last two reductions above are done exactly 
as in the paper of \cite{Sarma_2012} that show
the time lower bound of $\tilde{\Omega}(D+\sqrt{n})$ rounds; however, the
first reduction from SD to DSD is different in the smoothing model, as the
input graph used in the DSD  can use the (additional) power of the smoothing model. 

We will first briefly discuss the reductions as in 
\cite{Sarma_2012} and then discuss how to modify the first reduction to work for the smoothing model.
Here we first state
the bounds that we obtain via these reductions, and refer to \cite{Sarma_2012} for the details. The well-known communication complexity lower bound for the SD problem
is $\Omega(k)$ (see e.g.,\cite{Kushilevitz_1997_Book}), where $k$ is the length of the bit vector
of Alice and Bob. This lower bound holds even for randomized Monte-Carlo algorithms and even under {\em shared} public coin.

Due to the graph topology used in the DSD problem
(see Figure \ref{fig:lbgraph} without the smoothing edges) the value of  $k$ is to be $\Theta(\sqrt{n})$ (which is the best possible).
The reduction from SD to DSD shows that the lower
bound for DSD, is $\tilde{\Omega}(\sqrt{n})$ {\em rounds}.
(Note that the diameter of the lower bound graph is
$O(\log n)$ and hence subsumed.) This reduction
uses the {\em Simulation Theorem} (cf. Theorem 3.1
in \cite{Sarma_2012}) which is explained later below.

 The reduction from DSD to CCS 
 shows that the same time lower bound of $\tilde{\Omega}(\sqrt{n})$  rounds holds for the CCS problem. This reduction shows that the given subgraph $H$ in the  CCS problem is spanning connected if and only if
 the input vectors $x$ and $y$ are disjoint.

The reduction from CCS to MST problem shows that
the time lower bound for CCS verification which is
$\tilde{\Omega}(\sqrt{n})$ also holds for the MST problem.
This reduction takes as input the CCS problem and assigns weight 1 to the edges in the subgraph $H$ and weight $n$
to all other edges in $G$. It is easy to show
that $H$ is spanning connected if and only if the weight
of the MST is less than $n$. Hence the same time lower
bound that holds for CCS also holds for MST.


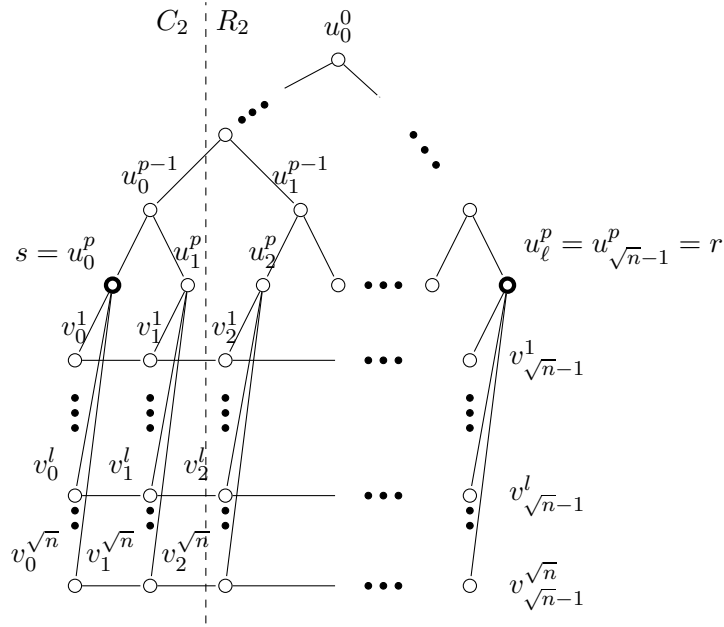
\begin{figure*} 
\centering
\begin{tikzpicture}[shorten >= 1 pt, -]
\tikzstyle{hollow}=[circle,draw,fill=white,minimum size=5pt,inner sep=0pt]
\tikzstyle{ring}=[circle,draw,ultra thick,fill=white,minimum size=5pt,inner sep=0pt]
\tikzstyle{dots}=[circle,draw,fill=black,minimum size=2.5pt, inner sep=0pt]
\tikzstyle{lineend}=[circle,draw,minimum size=0.001pt,inner sep=0pt]

\node (u00) at (0,0) [hollow,label=above:$u_0^0$]{};
\node (arbit1) at (-0.75,-0.4) [lineend]{};
\node (ucl1) at (-0.98,-0.6) [dots]{};
\node (ucl2) at (-1.14,-0.7) [dots]{};
\node (ucl3) at (-1.28,-0.8) [dots]{};
\node (arbit2) at (0.55,-0.5) [lineend]{};
\node (ucr1) at (1.0,-1.0)[dots]{};
\node (ucr2) at (1.15,-1.2)[dots]{};
\node (ucr3) at (1.3,-1.4)[dots]{};
\node (up2some) at (-1.5,-1) [hollow]{};
\node (up10) at (-2.5,-2) [hollow,label=above:$u_0^{p-1}$]{};
\node (up11) at (-0.5,-2) [hollow,label=above:$u_1^{p-1}$]{};
\node (up1some) at (1.75,-2) [hollow]{};
\node (up0) at (-3,-3) [ring,label=above left:{$s=u_0^p$}]{};
\node (up1) at (-2,-3) [hollow,label=above:$u_1^p$]{};
\node (up2)	at (-1,-3) [hollow,label=above:$u_2^p$]{};
\node (up3) at (0,-3) [hollow]{};
\node (upc1) at (0.4,-3) [dots]{};
\node (upc2) at (0.6,-3) [dots]{};
\node (upc3) at (0.8,-3) [dots]{};
\node (updp2) at (1.25,-3) [hollow]{};
\node (updp1) at (2.25,-3) [ring,label=above right:{$u^p_{\ell} = u^p_{\sqrt{n} - 1} = r$}]{};
\node (v10) at (-3.5,-4) [hollow,label=above:$v_0^1$]{};
\node (v11) at (-2.5,-4) [hollow,label=above:$v_1^1$]{};
\node (v12) at (-1.5,-4) [hollow,label=above:$v_2^1$]{};
\node (arbit3) at (0,-4) [lineend]{};
\node (v1dp1) at (1.75,-4) [hollow]{};
\node (v1dp1name) at (2,-4) [label=right:$v^1_{\sqrt{n} - 1}$]{};
\node (v1c1) at (0.4,-4) [dots]{};
\node (v1c2) at (0.6,-4) [dots]{};
\node (v1c3) at (0.8,-4) [dots]{};
\node (v2c1) at (-3.5,-4.5) [dots]{};
\node (v2c2) at (-2.5,-4.5) [dots]{};
\node (v2c3) at (-1.5,-4.5) [dots]{};
\node (v2cdp1) at (1.75,-4.5) [dots]{};
\node (v3c1) at (-3.5,-4.7) [dots]{};
\node (v3c2) at (-2.5,-4.7) [dots]{};
\node (v3c3) at (-1.5,-4.7) [dots]{};
\node (v3cdp1) at (1.75,-4.7) [dots]{};
\node (v4c1) at (-3.5,-4.9) [dots]{};
\node (v4c2) at (-2.5,-4.9) [dots]{};
\node (v4c3) at (-1.5,-4.9) [dots]{};
\node (v4cdp1) at (1.75,-4.9) [dots]{};
\node (vl0) at (-3.5,-5.8) [hollow,label=above left:$v_0^l$]{};
\node (vl1) at (-2.5,-5.8) [hollow,label=above left:$v_1^l$]{};
\node (vl2) at (-1.5,-5.8) [hollow,label=above left:$v_2^l$]{};
\node (arbit4) at (0,-5.8) [lineend]{};
\node (vlc1) at (0.4,-5.8) [dots]{};
\node (vlc2) at (0.6,-5.8) [dots]{};
\node (vlc3) at (0.8,-5.8) [dots]{};
\node (vldp1) at (1.75,-5.8) [hollow]{};
\node (vldp1name) at (2,-5.8) [label=right:$v^l_{\sqrt{n} - 1}$]{};
\node (vmc1) at (-3.5,-6) [dots]{};
\node (vmc2) at (-2.5,-6) [dots]{};
\node (vmc3) at (-1.5,-6) [dots]{};
\node (vmcdp1) at (1.75,-6) [dots]{};
\node (vnc1) at (-3.5,-6.2) [dots]{};
\node (vnc2) at (-2.5,-6.2) [dots]{};
\node (vnc3) at (-1.5,-6.2) [dots]{};
\node (vncdp1) at (1.75,-6.2) [dots]{};
\node (vT0) at (-3.5,-7) [hollow,label=above left:$v_0^{\sqrt{n}}$]{};
\node (vT1) at (-2.5,-7) [hollow,label=above left:$v_1^{\sqrt{n}}$]{};
\node (vT2) at (-1.5,-7) [hollow,label=above left:$v_2^{\sqrt{n}}$]{};
\node (arbit5) at (0,-7) [lineend]{};
\node (vTc1) at (0.4,-7) [dots]{};
\node (vTc2) at (0.6,-7) [dots]{};
\node (vTc3) at (0.8,-7) [dots]{};
\node (vTdp1) at (1.75,-7) [hollow]{};
\node (vTdp1name) at (2,-7) [label=right:$v^{\sqrt{n}}_{\sqrt{n} - 1}$]{};

\foreach \from/\to in {u00/arbit1,u00/arbit2,up2some/up10,up2some/up11,up10/up0,up10/up1,up11/up2,up11/up3,up1some/updp2,up1some/updp1,up0/v10,up1/v11,up2/v12,updp1/v1dp1,v10/v11,v11/v12,v12/arbit3,up0/vl0,up1/vl1,up2/vl2,updp1/vldp1,vl0/vl1,vl1/vl2,vl2/arbit4,up0/vT0,up1/vT1,up2/vT2,updp1/vTdp1,vT0/vT1,vT1/vT2,vT2/arbit5}
\draw [-, thin] (\from)--(\to);

\draw [dashed] (-1.76, -7.5)--(-1.76, 0.8);
\node[] at (-1.4, 0.5) {$R_2$};
\node[] at (-2.2, 0.5) {$C_2$};

\end{tikzpicture}
\caption{The lower bound graph.}
\label{fig:lbgraph}
\end{figure*}

\subsection{Reduction of SD to DSD: The Simulation Theorem}

In this section, we explain the key reduction from
SD to DSD which uses the {\em Simulation Theorem}  as in \cite{Sarma_2012}.
The reduction idea is as follows. Assuming that we have an algorithm for DSD that finished in $r$ rounds, Alice and Bob will simulate this algorithm in the two party model
by sending as few bits as possible. The Simulation theorem
accomplishes this. 
The Simulation theorem in \cite{Sarma_2012} et al that shows the simulation in the 
standard graph (without smoothing) uses constant number
of bits per round to do the simulation. Thus if the DSD algorithm finishes in $r$ rounds, then the Alice and Bob would have solved SD by exchanging $O(r)$ bits. If $r = o(k)$ (where $k$ is the length of the input bit string),
then this will contradict the lower bound of set disjointness in the two party model which is $\Omega(k)$
(even for randomized algorithms).

We note that the reduction is similar to that in \cite{Sarma_2012} as the input graph $G(x,y)$  for the DSD is the same (see Figure \ref{fig:lbgraph}).  However, in the smoothing model,
the algorithm has the additional power of using
the smoothing edges and hence the lower bound will be smaller as we will show below.
We  will first briefly describe the idea
behind the Simulation Theorem as it applicable to $G$
without the smoothing.  The  lower bound
graph $G(x,y)$ used in the Simulation Theorem is as follows.
Note that
the input graph $G$ for DSD has two distinguished nodes $s$ and $t$ that have the inputs $x$ and $y$ (corresponding to  Alice and Bob) respectively.

\subsubsection{The lower bound graph (family) for MST}

$G(x,y)$ is depicted as in Figure \ref{fig:lbgraph}, where $|x| = |y| = \sqrt{n}$. $G(.)$ includes $\sqrt{n}$ paths, and a {\em full binary tree} (to reduce the diameter to $O(\log n)$). Each {\em path} has length of $\sqrt{n}$; these are called the {\em path edges}. The full binary tree has $\sqrt{n}$ leaves, and hence, has the height of $p = \frac{\log{n}}{2}$. We number the leaves and the path nodes from left to right --- $0,1, \dots, \ell$.
Note that leaf numbered 0  is node $s$ and leaf numbered $\ell$ is $t$.   Consider each leaf $0< j < \ell$, let it connect to all the nodes $j$ in all paths, we call these {\em spoke} edges. Note that the binary tree edges, the 
path edges, and the spoke edges from leaf nodes other
than $s$ and $t$ are present in every graph
of the family.

It is straightforward that $G(x,y)$ has $\Theta(n)$ nodes, and diameter $D = \Theta(\log{n})$.

In the reduction from SD to DSD, Alice and Bob wants
to solve SD problem with Alice having input vector
$x =(x_1,x_2,\dots, x_k)$ and Bob having input
vector $y=(y_1,y_2,\dots, y_k)$ where we fix $k = \sqrt{n}$. Depending on $x$ and $y$, Alice
and Bob will fix the spoke edges from nodes $s$ and
$t$. For the SD problem, Alice  will add edges from $s$ to
the first node in path $i$ if and only if $x_i =0$. Similarly Bob will add edges from $t$ to the last node in the path $i$ if and only if $y_i$ is 0.

Now we are ready to describe the Simulation Theorem
which really gives an algorithm for Alice and Bob
to simulate any given algorithm for DSD problem.
If the DSD algorithm runs in $r$ rounds, the Simulation theorem will show how to solve the SD problem by exchanging $O(r)$ bits. We sketch the main idea here, which is quite simple, and leave the
 full details which can be found in \cite{Sarma_2012}.

How will Alice and Bob start the simulation? Note
that they want to solve the SD problem on their respective
inputs $x$ and $y$ in the two-party model. They will then use their respective inputs to construct $G(x,y)$ as described above. Note that Alice will be able to construct
all edges of $G$ except the spoke edges of $t$ (since that depends on Bob's input)  and vice versa for Bob. Then, assuming that there is an algorithm for the DSD problem on the same inputs that runs in $r$ rounds, Alice and Bob will simulate the DSD algorithm whose output will also give the output for the SD problem (by definition).

The main idea is for Alice and Bob to keep the simulation going as long as possible. If one {\em disregards the binary tree edges}, all paths are of length $\ell = \Theta(\sqrt{n})$
and hence it is easy to keep the simulation going
for $\ell/2$ rounds (say). This is because, Alice has all the information needed to simulate the DSD algorithm till
$\ell/2$ steps (i.e., the middle
of the path). Why? Because Alice knows her own input
and all other nodes in $G$ does not have any input. She does not know Bob's input, but does not matter for $\ell/2$ steps since in so many rounds nothing from
Bob's part of the graph reaches the ``middle" of the path.
But of course, the above is not true because of the binary tree edges which has smaller diameter. So to keep the simulation going for Bob, Alice sends
the minimum amount of information needed by him. Note that after $i$ rounds
the computation from Alice's side (which are the set
of nodes numbered $u_0^p$ and $v_0^1, v_0^2, \dots, v_0^{\sqrt{n}}$ will have reached nodes at distance $i$ on the path.

We define the $R_i$ (intuitively $i$-Right) set as follows: $R_i$ includes all nodes on the paths with subscript $j \geq i$, all leaf nodes $u_j^p$ where $j \geq i$, and all ancestor of these leaves, see Figure \ref{fig:lbgraph}.

In round $i$, Bob needs to keep the correct computation for $R_i$. To achieve that, Alice sends
{\em only the messages sent by the tree nodes in $R_{i-1}$ crossing into the tree nodes in $R_i$} (see Lemma 3.4 in \cite{Sarma_2012}). (A similar observation applies for Alice to do her simulation).
Hence only messages sent by at most $O(\log n)$ nodes in the binary tree are needed. Hence in every round
at most $O(\log^2 n)$ bits need to be exchanged by Alice and Bob to keep the simulation going. Thus if 
the DSD algorithm finishes in $o(\ell/\log^2n) = o(\sqrt{n}/\log^2 n)$
rounds, then the simulation will also end successfully.

Thus we can show the following Simulation Theorem.

\begin{theorem}[Simulation Theorem] (Simplified version of Theorem 3.1 in \cite{Sarma_2012})
\label{theorem:base-simulation}
Given the DSD problem with input size $\Theta(\sqrt{n})$, encoded as $G(x,y)$ (i.e., the $x$ and $y$ are bit vectors of length
$\Theta(\sqrt{n})$), if there is a distributed algorithm  that solves the DSD problem in time at most $T$ rounds, using messages of size $O(\log n)$. Then there is an algorithm in the two-party communication complexity model that decides SD problem while exchanging at most $O(T\log^2{n})$ bits.
\end{theorem}

\subsection{Lower bound with smoothing}
Under the $\epsilon$-smoothing model, we will use the same lower bound graph $G(x,y)$. However, the smoothing model gives the algorithm additional power to add random edges in $G(x,y)$ during the course of the simulation. 
We will show how to modify the Simulation Theorem to apply
to the smoothing model. Naturally, since the algorithm has additional power, it can finish faster, and hence
the corresponding lower bound will be smaller.

We focus on Bob and show how he can keep his simulation going. (A similar argument applies for Alice.)
We consider how Bob maintains correct states for 
$R_i$ in round $i$. 
Beside the required messages as discussed in  Theorem \ref{theorem:base-simulation}, Bob needs to know the messages sent over (potential) {\em smoothed edges} crossing $V \setminus R_i$ into $R_i$. Since there are more messages to keep track, Bob cannot keep the simulation longer than that of the non-smoothing case.  Let $\delta$ be the number of rounds the simulation is valid without exceeding $\Theta(\sqrt{n})$ bits of communication. Since the smoothing process is randomized, we bound $\delta$ in expectation.

We will use a  pessimistic estimation to estimate
the number of smoothed edges that ``affect" Bob from
Alice's side. Let $C_i = V \setminus R_i$, let $S_i$ be the expected number of smoothed edges crossing $C_i$ and $R_i$. $S_{i}$ indicates the number of extra messages Bob needs to know, in round $i$. Notice that $|C_i| = i \Theta(\sqrt{n})$. 

\begin{equation}
    S_{i} = 2 \epsilon i \Theta(\sqrt{n})\frac{(n - i \Theta(\sqrt{n}))}{n}
\end{equation}

Since $\delta < \Theta(\sqrt{n})$, for every round $i$, $S_i = \Theta(\epsilon\delta\sqrt{n})$. After $\delta$ rounds, the expected number of messages over the smoothed edges is thus: $\Theta(\epsilon \delta^2 \sqrt{n})$. Also, by Theorem \ref{theorem:base-simulation}, we need to keep track of $(\delta \log{n})$ messages. To stay within the budget of $\sqrt{n}$ bits for DSD communication, we require:
\begin{equation*}
    \epsilon \delta^2 \Theta(\sqrt{n}) + \delta \Theta(\log{n}) \leq \frac{\sqrt{n}}{B}.
\end{equation*}

With $B = \Theta(\log{n})$ (the message size), we have $\delta \leq\Theta(\frac{1}{\sqrt{\epsilon \log{n}}}) $. 
Thus, we can keep the simulation going for up to
$ \Theta(\frac{1}{\sqrt{\epsilon \log{n}}})$ rounds.

To complete the lower bound argument, the same
lower bound applies for MST  in the $\epsilon$-smoothing model by the chain of reductions.


\section{Other Smoothing Models} \label{section-other-smoothing-models}

In this section, we discuss some of the other plausible smoothing models for the distributed MST problem. 
The most natural smoothing model that comes into the mind first in respect to a numerical-valued computing problem and that is similar in spirit to the original Spielman-Teng smoothing \cite{Spielman_2004} is where one ``perturbs'' the edge-weights in the given input graph. However, as we have already noted (see Section \ref{section-introduction}) this may not make for anything interesting if the perturbations are too small with respect to the original edge-weights.

One is therefore tempted to make the perturbations \emph{large} relative to the values of the original edge-weights. This, however, (depending on exactly how many edges are thus smoothed) may have the effect of essentially producing an input graph where the edge-weights are \emph{randomly distributed}, i.e., where every edge-weight is chosen uniformly at random from a certain range of values. As has been shown previously \cite{Khan_2008}, an MST can be constructed for such an input graph rather fast (in $\tilde{O}(D)$ rounds, actually, where $D$ is the diameter of the input graph).\footnote{The number of edges whose weights are smoothed can be parameterized, however. For example, one can consider a smoothing model where each edge in the original input graph is smoothed (i.e., its weight is perturbed) with a certain probability $\epsilon$, say, for some small $\epsilon$ (e.g., $\epsilon$ can be made $\tilde{\Theta}(\frac{1}{\sqrt{n}})$).} Thus, while this model with large perturbations in edge-weights is more interesting than the smoothing model with small perturbations in edge-weights, it has been somewhat --- even if not fully --- explored in the distributed MST literature.

These considerations motivate us to explore an alternate avenue where the \emph{graph topology} rather than the edge-weights are perturbed. In particular, we consider smoothing models where the perturbation process adds more edges to the original input graphs.
Below we describe two such models where additional edges characterize the perturbed graph compared to the original input graph. We call these models the \emph{$k$-smoothing models}. We distinguish these two models based on whether or not the smoothed edges are known to the algorithm.

\begin{remark}
    The usual smoothing concept (e.g., as originally put forth by Spielman and Teng \cite{Spielman_2004}, and as recently applied in the context of dynamic networks \cite{Dinitz_2018}) dictates that the new, ``perturbed'' graph be chosen according to some distribution (usually the uniform distribution) from a set of graphs that are ``close'' to the original input graph according to some distance metric. Thus the new, perturbed graph is presented to the algorithm as a whole and the algorithm has no way of knowing which edges are additional and which edges are original.

    Our $\epsilon$-smoothing model, however, focused on the other scenario (where the algorithm does have such knowledge)  because (1) it is algorithmically easier to approach and (2) also it may be relevant in somewhat different contexts too (e.g., in the context of \emph{congested clique} as discussed in Section \ref{section-introduction}).
\end{remark}
\subsection{The $k$-smoothing model with known smoothed edges}

The main model that we follow in this paper, the $\epsilon$-smoothing model (see Section \ref{section-computing-model-main-smoothing-model}) adds smooth edges by node-local computation during the course of an algorithm. We can look at models where smooth edges are added all at once, by some external process, prior to the commencement of the algorithm.

\begin{enumerate}

    \item Consider a smoothing model $\mathcal{M}(k, *)$ where the ``perturbation'' process \emph{adds}
    $k$ additional edges to the original input graph. In our notation, $*$ denotes the fact that the smoothed (i.e., additional) edges are known to the algorithm.

    These $k$ additional edges are chosen uniformly at random from all the $\binom{n}{2}$ possible edges of the graph. 
  
    \item Consider another smoothing model $\mathcal{M}(\delta, *)$ where the ``perturbation'' process \emph{adds} --- independently for each of the $\binom{n}{2}$ possible edges --- an edge with probability $\delta$.

    
\end{enumerate}

\begin{remark}
    It is not difficult to see that the model $\mathcal{M}(\epsilon, *)$ is essentially equivalent to the  model $\mathcal{M}(k, *)$ for the case when $\epsilon = \frac{k}{n}$.
\end{remark}

\begin{remark}
    By Lemma \ref{lemma:smooth-gnp}, consider the $\epsilon$-smoothing model with $\ell$ rounds of smoothing, then: $\mathcal{M}(k=2\ell\epsilon n, *)$ and $\mathcal{M}(\delta=\frac{2\ell\epsilon}{n}, *)$ are the equivalent models.
\end{remark}
\subsection{The $k$-smoothing model with unknown smoothed edges}

Essentially, we have the counterparts of $\mathcal{M}(k, *)$ and $\mathcal{M}(\delta, *)$ for the particular case when the smoothed edges are \emph{not} known to the algorithm. We denoted these models $\mathcal{M}(k, \times)$ and $\mathcal{M}(\delta, \times)$.

    We also note that both $\mathcal{M}(k, \times)$ and $\mathcal{M}(\epsilon, \times)$ are essentially equivalent to the smoothing model proposed by Dinitz et al.\ \cite{Dinitz_2018}, where the new, perturbed input graph is chosen uniformly at random from the set of all possible graphs whose edge-sets are at a (positive) edit-distance $k$ from the original input graph. The only subtle difference is that, in their model \cite{Dinitz_2018}, the edit-distance can be positive as well as negative. We, however, consider only \emph{positive} edit-distances here.

We note that the algorithms specified in this paper do not (at least directly) work when the smoothing edges are not known. However, note
that the lower bound holds (for appropriate choice of $k$ in terms of $\epsilon$).

\section{Conclusion}

In this paper, we study smoothed analysis of distributed graph algorithms focusing on the well-studied distributed MST problem. 
Our work can be considered as a first step in understanding the smoothed complexity of distributed graph algorithms.

We present a smoothing model, and upper and lower bounds for the time complexity of distributed MST in this model. These bounds quantify the time bounds in terms of the smoothing parameter $\epsilon$. The bounds are within a factor of $2^{O(\sqrt{\log n})} \polylog n$ and a key open problem is whether this gap can be closed.

While we focus on one specific smoothing model, our results also apply to other related smoothing models (discussed in Section \ref{section-other-smoothing-models}). A commonality among these models, besides adding random edges, is that the added edges are \emph{known} to the nodes. This knowledge of the random edges are crucial to obtaining our upper bounds. Of course, our lower bounds apply \emph{regardless of} this knowledge.

An important open problem is to show non-trivial bounds when the random edges are {\em unknown} to the nodes; i.e., the input graph consists of the original graph $G$ plus the random edges and the nodes cannot distinguish  between edges in $G$ and the added random edges. 

It would also be interesting to explore other fundamental distributed graph problems such as leader election, shortest paths, minimum cut etc., in the smoothing model.

\bibliographystyle{plain}
\bibliography{references_MST_smoothed_analysis}

\clearpage
\appendix

\section{Preliminaries}

\subsection{Volume and Graph Conductance}

For any graph $G = (V, E)$ and any node $v \in V$, the degree of $v$ is denoted by $d(v)$. For any nonempty set $S \subseteq V$, we define the following terms.

\begin{definition} \label{definition-volume-of-a-set}
    The \emph{volume} of $S$ is defined as the sum of the degrees of nodes in $S$, i.e., the volume of $S$, $\Vol(S)  \defeq  \sum_{v \in S} d(v)$.
\end{definition}

\begin{definition} \label{definition-boundary-of-a-set}
    For $S  \subsetneq  V$, the \emph{boundary} of $S$, denoted by $\partial{S}$, is defined as the set of edges crossing $S$ and $\bar{S}$, where $\bar{S}  \defeq  V \setminus S$.
\end{definition}

\begin{definition} \label{definition-conductance-of-a-set}
    For $S  \subsetneq  V$, the \emph{conductance} of $S$ is defined as $\varphi(S)  \defeq  \frac{\SizeOfSet{\partial{S}}}{\min(\Vol(S), \Vol(\bar{S}))}$.
\end{definition}

\begin{definition} \label{definition-conductance-of-the-whole-graph}
    The \emph{conductance} of $G$ is defined as
    \begin{equation}
        \Phi(G)  \defeq  \min \Set{\varphi(S)  |  S \subsetneq V}
    \end{equation}
\end{definition}

\begin{remark}
    An alternative, but equivalent, definition of $\Phi(G)$ would be the following:
    \begin{equation} \label{equation-alternative-definition-conductance-of-the-whole-graph}
        \Phi(G)  \defeq  \min \Set{\frac{\SizeOfSet{\partial{S}}}{\Vol(S)}  |  S \subsetneq V \text{ and } \Vol(S) \leq \frac{\Vol(V)}{2}}
    \end{equation}
\end{remark}

\begin{remark}
Another equivalent definition of $\Phi(G)$, which is more handy in some cases:
\begin{equation}
\label{equation:conductance-handy}
    \Phi(G)  \defeq 
        \min \Set{
            \frac{\SizeOfSet{\partial{S}}}{min(\Vol(S),\Vol(\bar{S}))}  |  S \subsetneq V \text{ and } |S| \leq \frac{|V|}{2}}
\end{equation}
\end{remark}
\subsection{Random Walks and Mixing Time}

Our MST algorithms (see Section \ref{section-constructing-an-MST} and Section \ref{section-improved-algorithm}) use Ghaffari et al.'s works (see \cite{Ghaffari_2017_PODC, Ghaffari_2018_DISC}) to perform \emph{routing}. They use \emph{random walks} extensively. We follow their definition of \emph{mixing time} and state it here for the sake of completeness. Note that the usual definitions most commonly used in the literature are \emph{conceptually} the same.

\paragraph{Random Walks.} Random walks have turned out to be a fundamental technique that continues to find applications in all branches of theoretical computer science, and in particular, distributed computing (see, e.g., \cite{Sarma_2013, Sarma_2015}). In order to guarantee that the resulting Markov chain is aperiodic, Ghaffari et al.\ uses \emph{lazy random walks}: In every step, the walk remains at the current node with probability $\frac{1}{2}$ and it transitions to a uniformly random neighbor otherwise. The \emph{stationary distribution} of such a random walk is proportional to the degree distribution, i.e., when performing enough steps of the random walk, the probability for ending in a node v converges to $\frac{d(v)}{2m}$.

\begin{definition}[Mixing Time] \label{definition-mixing-time}
    For $V = \Set{v_1, v_2, \ldots, v_n}$ and a node $v \in V$, let $P_v^t  =  (P_v^t(v_1), P_v^t(v_2), \ldots, P_v^t(v_n))$ be the probability distribution on the nodes after $t$ steps of a lazy random walk starting at $v$. Then the \emph{mixing time} (denoted by $\tau_{\text{mix}}(G)$) of the graph $G$ is defined as the minimum $t$ such that $\forall\ u, v \in V$,
    \begin{center}
        $\abs{P_v^t(u) - \frac{d(v)}{2m}}  \leq   \frac{d(v)}{2mn}$.
    \end{center}
\end{definition}

We note that --- as in \cite{Ghaffari_2017_PODC} --- by running a random walk for $O(\tau_{\text{mix}})$ steps, one can improve the deviation from mixing time to $\abs{P_v^t(u) - \frac{d(v)}{2m}}  \leq   \frac{1}{n^c}$, for any arbitrary (but fixed) positive constant $c$.
\subsection{Chernoff Bounds}

Here we state the most commonly used version of the Chernoff bound \cite{Mitzenmacher_2017_Book}: for the \emph{tail distribution} of a sum of independent, $0-1$ random variables. Let $X_1, X_2, \ldots, X_n$ be a sequence of independent, $0-1$ random variables with $\Pr[X_i = 1]  =  p_i$. Let $X  =  \sum_{i = 1}^n X_i$, and let

\begin{align*}
    &\mu  =  \E[X]  =  \E[\sum_{i = 1}^n X_i]  =  \sum_{i = 1}^n \E[X_i]  =  \sum_{i = 1}^n p_i\text{.}
\end{align*}

For a given $\delta > 0$, we are interested in bounds on $\Pr[X  \geq  (1 + \delta)\mu]$ and $\Pr[X  \leq  (1 - \delta)\mu]$ --- that is, the probability that $X$ deviates from its expectation $\mu$ by $\delta \mu$ or more. Then one can show that (see \cite[Chapter 4]{Mitzenmacher_2017_Book} for the proof of the theorem and a more detailed discussion):

\begin{theorem}[Chernoff bounds] \label{theorem-Chernoff-bounds}
    Let $X_1, X_2, \ldots, X_n$ be a sequence of independent, $0-1$ random variables such that $\Pr[X_i = 1]  =  p_i$. Let $X  =  \sum_{i = 1}^n X_i$ and let $\mu  =  \E[X]$. Then the following Chernoff bounds hold:

    \begin{enumerate}
        \item \label{clause-inequality-Chernoff-bound-upper-tail} For any $\delta  \in  (0, 1]$,
        \begin{center}
            $\Pr[X  \geq  (1 + \delta)\mu]   \leq   \exp(-\frac{\mu \delta^2}{3})$;
        \end{center}

        and

        \item \label{clause-inequality-Chernoff-bound-lower-tail} For any $\delta  \in (0, 1)$,
        \begin{center}
            $\Pr[X  \leq  (1 - \delta)\mu]   \leq   \exp(-\frac{\mu \delta^2}{2})$.
        \end{center}
    \end{enumerate}

\end{theorem}

\begin{corollary}[Chernoff bound]
    Let $X_1, X_2, \ldots, X_n$ be a sequence of independent, $0-1$ random variables such that $\Pr[X_i = 1]  =  p_i$. Let $X  =  \sum_{i = 1}^n X_i$ and let $\mu  =  \E[X]$. Then, for $0 < \delta < 1$,
    \begin{center}
        $\Pr[\abs{X - \mu}   \geq   \delta\mu]     \leq     2\exp(-\frac{\mu \delta^2}{3})$.
    \end{center}
\end{corollary}

\begin{remark}
    In practice, we often do not have the exact value of $\E[X]$. Instead we can use $\mu  \geq  \E[X]$ in Clause \ref{clause-inequality-Chernoff-bound-upper-tail} and $\mu  \leq  \E[X]$ in Clause \ref{clause-inequality-Chernoff-bound-lower-tail} of Theorem \ref{theorem-Chernoff-bounds}.
\end{remark}

\section{Alternative Proof for Lemma \ref{lemma:conductance}} \label{section-proof-of-the-expansion-and-conductance-lemma}

\begin{proof}
  We will show that, there exists constants $c > 0$ and $0 < \alpha, \beta < 1$, such that:
  $\Phi(R(G)) \geq \frac{c(1-\beta)(1-\alpha)}{4(1+\alpha)}$.

  Let the selection run for $c \frac{\log n}{\epsilon}$ rounds. For a node $u$, let $r(u)$ indicate
  the number of added edges by $u$. We have $\E[r(u)] = c \log n$. We will show that, with high probability:
  \begin{equation} \label{eq:bound-ru}
    \forall u \in V, (1-\alpha)c \log n \leq r(u) \leq (1+\alpha)c \log n
  \end{equation}

  Let $E_1$ be the bad event, that at least one  node added too few or too many edges
  outside the above range. Using a standard Chernoff bound to bound the probability that, for one particular node
  $u$, $r(u)$ falls outside the above range and then
  union-bounding over $n$ nodes, we have:
  \begin{align}
    \Pr[E_1] &\leq n\left( \Pr[r(u) \leq (1-\alpha)c\log n] +
      \Pr[r(u) \geq (1+\alpha)c\log n] \right) \nonumber \\
    &\leq n \left( e^{-\frac{\alpha^2 c\log n}{3}} + e^{-\frac{\alpha^2 c\log n}{2}} \right)
    \nonumber \\
    &\leq exp \left( \ln n + \ln 2 - \frac{\alpha^2 c\log n}{3} \right) = o(n^{-1}),
  \end{align}
  for  suitable constants $\alpha$ and $c$.
  
  We note that $R(G)$ is essentially equivalent to
  the Erdos-Renyi random graph $G(n,p)$ where $p = \Theta(\log n/n)$ is the probability of having an edge between any pair of nodes. Note that, however, there are slight differences ---  $R(G)$ is a multi-graph unlike $G(n,p)$. Since $G(n,p)$ with $p=\Theta(\log n/n)$ is known to be an expander and hence has constant conductance and $O(\log n)$ mixing time, here we give a self-contained proof for the conductance 
  
  Now we bound the conductance $\Phi(R)$ of the graph $R(G)$, using the definition in equation \ref{equation:conductance-handy}:
  \begin{equation}
    \Phi(R(G)) =  \min_{S \subset V, |S| \leq \frac{n}{2}}
    \frac{|\partial S|}{\min(\Vol(S), \Vol(\bar S))}
  \end{equation}
  Fix a subset $S$ of $V$ where $|S| = k \leq \frac{n}{2}$, and bound the conductance $\varphi(S)$:
    \begin{equation}
    \varphi(S) = \frac{|\partial S|}{\min(\Vol(S), \Vol(\bar S))}
  \end{equation}
 
  We consider the numerator and the denominator of $\varphi(S)$ separately. For the denominator, we have: $\Vol(S)
  \leq 2k(1+\alpha)c\log n$ and $\Vol(\bar S) \leq 2(n-k)(1+\alpha)c\log n$. Thus:
  \begin{equation}
    \min(\Vol(S), \Vol(\bar S)) \leq 2k(1+\alpha)c\log n \label{eq:bound-denom}
  \end{equation}
  Consider the numerator, and recall that each node $u$ has $r(u)$ edges where the destination of each edge is
  chosen uniformly at random.  Using the bound in equation \ref{eq:bound-ru}, we have:
  \begin{align}
    E[|\partial S|] &\geq \min \left(
      k(1-\alpha)c\log n \cdot (\frac{n-k}{n}),
      (n-k)(1-\alpha)c\log n \cdot (\frac{k}{n})
    \right) \nonumber \\
    &= \frac{n-k}{n}k(1-\alpha)c \log n
  \end{align}
  Since $k \leq \frac{n}{2}$, we have:
  \begin{equation}
    E[|\partial S|] \geq \frac{k}{2}(1-\alpha)c \log n
  \end{equation}

  Let $\mu_{S} = E[|\partial S|]$. Next, similar to the steps for equation \ref{eq:bound-ru},
  we will bound $|\partial S|$. For some constant $0 < \beta < 1$, we need to show that, with high
  probability:
  \begin{equation}
    |\partial S| \geq (1 - \beta) \mu_S \label{eq:bound-numer}
  \end{equation}
  Chernoff bound gives:
  \begin{equation}
    \Pr[|\partial S| < (1 - \beta) \mu_S] \leq e^{-\frac{\beta^2 \mu_S}{2}}
  \end{equation}
  Let $\mathcal{E}^k_2$ be the bad event that at least one of $\binom{n}{k}$ set has less than
  $(1-\beta)\mu_S$ boundary edges. Using the fact (and $k\geq 1$):
  \[
    \left(\frac{n}{k} \right)^k \leq {\binom{n}{k}} \leq \left( \frac{ne}{k} \right)^k
    \leq (e^{k \ln n}), 
  \]
  and union bound gives:
  \begin{align}
    \Pr[\mathcal{E}^k_2] \leq exp(-\frac{\beta^2 \mu_S}{2} + k\ln n)
    \leq \exp(-\frac{\beta^2 k}{4} (1-\alpha)c \log n + k \ln n)
  \end{align}
  We can choose the constants $c, \alpha, \beta$ such that:
  \begin{equation}
    \Pr[\mathcal{E}^k_2] \leq e^{-2\ln n}
  \end{equation}
  Let $\mathcal{E}_3$ be the bad event that at least one bad event $\mathcal{E}^k_2$ happens, for all
  $1 \leq k \leq \frac{n}{2}$. Using union bound:
  \begin{align}
    \Pr[\mathcal{E}_3] &\leq \sum_{k=1}^{n/2} \Pr[\mathcal{E}^k_2]
    \leq e^{\ln n - \ln 2} e^{-2 \ln n} = o(n^{-1})
  \end{align}
  Thus, from equations \ref{eq:bound-denom} and \ref{eq:bound-numer}, for
  all $S$ where $|S| \leq \frac{n}{2}$, we have, with high probability (for suitably chosen constants $c, \alpha, \beta$):
  \begin{equation}
    \varphi(S) \geq \frac{(1-\beta)\frac{k}{2}(1-\alpha)c \log n}
      {2k(1+\alpha)c \log n}
    = \frac{c(1-\beta)(1-\alpha)}{4(1+\alpha)}
  \end{equation}
  It follows that $\Phi(R(G))$ is constant.

  Since $R(G)$ has constant conductance, it follows that (see, e.g., \cite{Jerrum_1988, Lovasz_1999_STOC}) $R(G)$ has fast mixing time: $\tau_{mix}(R)=O(\log{n})$.
\end{proof}

\end{document}